%% file: fscd20.tex
\documentclass[a4paper,UKenglish,cleveref,autoref,thm-restate]{lipics-v2019}
\title{Type safety of rewrite rules in dependent types} %TODO Please add

\author{Fr\'ed\'eric Blanqui}{Universit\'e Paris-Saclay, ENS Paris-Saclay, CNRS, Inria\and
  Laboratoire Sp\'ecification et V\'erification, 94235, Cachan, France}{}{https://orcid.org/0000-0001-7438-5554}{}%TODO mandatory, please use full name; only 1 author per \author macro; first two parameters are mandatory, other parameters can be empty. Please provide at least the name of the affiliation and the country. The full address is optional
\authorrunning{F. Blanqui} %TODO mandatory. First: Use abbreviated first/middle names. Second (only in severe cases): Use first author plus 'et al.'

\Copyright{Inria} %TODO mandatory, please use full first names. LIPIcs license is "CC-BY";  http://creativecommons.org/licenses/by/3.0/

 %TODO mandatory; please add comma-separated list of keywords

\relatedversion{} %optional, e.g. full version hosted on arXiv, HAL, or other respository/website
\supplement{\url{https://github.com/wujuihsuan2016/lambdapi/tree/sr}}%optional, e.g. related research data, source code, ... hosted on a repository like zenodo, figshare, GitHub, ...

\acknowledgements{The author thanks Jui-Hsuan Wu for his prototype implementation and his comments on a preliminary version of this work.}%optional

\nolinenumbers %uncomment to disable line numbering

\EventEditors{Zena M. Ariola}
\EventNoEds{1}
\EventLongTitle{5th International Conference on Formal Structures for Computation and Deduction (FSCD 2020)}
\EventShortTitle{FSCD 2020}
\EventAcronym{FSCD}
\EventYear{2020}
\EventDate{June 29--July 5, 2020}
\EventLocation{Paris, France}
\EventLogo{}
\SeriesVolume{167}
\ArticleNo{11}
\newcommand\hide[1]{}

\newcommand\D\Delta
\newcommand\G\Gamma
\newcommand\s\sigma
\newcommand\T\Theta
\newcommand\vep\varepsilon

\newcommand\B\Box
\newcommand\all\forall

\newcommand\sle\subseteq

\newcommand\tgt\rhd

\newcommand\mc\mathcal
\newcommand\mr\mathrm
\newcommand\mi\mathit

\newcommand\cD{\mc{D}}
\newcommand\cE{\mc{E}}
\newcommand\cF{\mc{F}}
\newcommand\cR{\mc{R}}
\newcommand\cS{\mc{S}}
\newcommand\cT{\mc{T}}
\newcommand\cV{\mc{V}}

\newcommand\FV{\mr{FV}}
\newcommand\Pos{\mr{Pos}}

\newcommand\h[1]{{\widehat{#1}}}

\newcommand\ie{{\em i.e.} }

\input{dedukti}

\begin{document}

\maketitle

\begin{abstract}
\input{abstract}
\end{abstract}

\input{paper}

\bibliography{fscd20}

\end{document}

%% file: dedukti.tex
\usepackage[utf8]{inputenc}

\usepackage{xcolor}
\definecolor{green}{RGB}{0,130,0}
\definecolor{lightgrey}{RGB}{240,240,240}

\usepackage{amssymb}

\usepackage{listings}

\lstdefinelanguage{Dedukti}
{
  inputencoding=utf8,
  extendedchars=true,
  escapechar=`,
  numbers=none,
  numberstyle={},
  tabsize=2,
  basicstyle={\ttfamily\small\upshape},
  backgroundcolor=\color{lightgrey},
  keywords={abort,admit,apply,as,assert,assertnot,assume,compute,constant,definition,focus,in,injective,let,open,print,private,proof,proofterm,protected,qed,refine,reflexivity,require,rewrite,rule,set,simpl,symmetry,symbol,theorem,type,TYPE,with,why3},
  sensitive=true,
  keywordstyle=\color{blue},
  morecomment=[l]{//},
  commentstyle={\itshape\color{red}},
  string=[b]{"},
  stringstyle=\color{orange},
  showstringspaces=false,
  literate={Π}{$\Pi$}1 {→}{$\rightarrow$}1 {λ}{$\lambda$}1
  {≔}{$\coloneqq$}1 {↪}{$\hookrightarrow$}1
  {∧}{$\wedge$}1 {⇒}{$\Rightarrow$}1 {∀}{$\forall$}1
  {𝔹}{$\mathbb{B}$}1 {𝕃}{$\mathbb{L}$}1 {ℕ}{$\mathbb{N}$}1
  {α}{$\alpha$}1 {β}{$\beta$}1 {π}{$\pi$}1 {τ}{$\tau$}1 {ω}{$\omega$}1
  {≤}{$\le$}1 {≠}{$\neq$}1 {∉}{$\notin$}1
}

%% file: abstract.tex
The expressiveness of dependent type theory can be
extended by identifying types modulo some additional computation rules. But, for
preserving the decidability of type-checking or the logical
consistency of the system, one must make sure that those user-defined
rewriting rules preserve typing. In this paper, we give a new
method to check that property using Knuth-Bendix completion.

%% file: paper.tex
%%% Local Variables:
%%%   mode: flyspell
%%%   ispell-local-dictionary: "english"
%%% End:

\section{Introduction}

The $\lambda\Pi$-calculus, or LF \cite{harper93jacm}, is an extension of
the simply-typed $\lambda$-calculus with dependent types, that is, types
that can depend on values like, for instance, the type $Vn$ of vectors
of dimension $n$. And two dependent types like $Vn$ and $Vp$ are
identified as soon as $n$ and $p$ are two expressions having the same
value (modulo the evaluation rule of $\lambda$-calculus, $\beta$-reduction).

In the $\lambda\Pi$-calculus modulo rewriting, function and type symbols
can be defined not only by using $\beta$-reduction but also by using
rewriting rules \cite{terese03book}. Hence, types are identified
modulo $\beta$-reduction and some user-defined rewriting rules. This
calculus has been implemented in a tool called Dedukti
\cite{dedukti}.

Adding rewriting rules adds a lot of expressivity for encoding logical
or type systems. For instance, although the $\lambda\Pi$-calculus has no
native polymorphism, one can easily encode higher-order logic or the
calculus of constructions by using just a few symbols and rules
\cite{cousineau07tlca}. As a consequence, various tools have been
developed for translating actual terms and proofs from various systems
(Coq, OpenTheory, Matita, Focalize, \ldots) to Dedukti, and back,
opening the way to some interoperability between those systems
\cite{assaf19draft}. The Agda system recently started to experiment
with rewriting too \cite{cockx16types}.

To preserve the decidability of type-checking and the logical
consistency, it is however essential that the rules added by the user
preserve typing, that is, if an expression $e$ has some type $T$ and
a rewriting rule transforms $e$ into a new expression $e'$, then $e'$
should have type $T$ too. This property is also very important in
programming languages, to avoid some errors (a program declared to
return a string should not return an integer).

When working in the simply-typed $\lambda$-calculus, it is not too
difficult to ensure this property: it suffices to check that, for
every rewriting rule $l\hookrightarrow r$, the right-hand side (RHS) $r$
has the same type as the left-hand side (LHS) $l$, which is decidable.

The situation is however much more complicated when working with
dependent types modulo user-defined rewriting rules. As type-checking
requires one to decide the equivalence of two expressions, it is
undecidable in general to say whether a rewriting rule preserves
typing, even $\beta$-reduction alone \cite{saillard15phd}.

Note also that, in the $\lambda\Pi$-calculus modulo rewriting, the set of
well-typed terms is not fixed but depends on the rewriting rules
themselves (it grows when one adds rewriting rules).

Finally, the technique used in the simply-typed case (checking that
both the LHS and the RHS have the same type) is not satisfactory in
the case of dependent types, as it often forces rule left-hand sides
to be non-linear \cite{blanqui05mscs}, making other important
properties (namely confluence) more difficult to establish and the
implementation of rewriting less efficient (if it does not use
sharing).

\begin{example}\label{ex-beta}
  As already mentioned, Dedukti is often used to encode logical
  systems and proofs coming from interactive or automated theorem
  provers. For instance, one wants to be able to encode the
  simply-typed $\lambda$-calculus in Dedukti. Using the new Dedukti
  syntax\footnote{\url{https://github.com/Deducteam/lambdapi}}, this
  can be done as follows (rule variables must be prefixed by {\tt \$}
  to distinguish them from function symbols with the same name):

  \begin{lstlisting}
constant symbol T: TYPE // Dedukti type for representing simple types
constant symbol arr: T → T → T // arrow simple type constructor

injective symbol τ: T → TYPE// interprets T elements as Dedukti types 
rule τ (arr $x $y) ↪ τ $x → τ $y // (Curry-Howard isomorphism)

// representation of simply-typed λ-terms
symbol lam: Π a b, (τ a → τ b) → τ (arr a b)
symbol app: Π a b, τ (arr a b) → (τ a → τ b)

rule app $a $b (lam $a' $b' $f) $x ↪ $f $x // β-reduction
  \end{lstlisting}

  Proving that the above rule preserves typing is not trivial as it is
  equivalent to proving that $\beta$-reduction has the
  subject-reduction property in the simply-typed
  $\lambda$-calculus. And, indeed, the previous version of Dedukti was
  unable to prove it.

  The LHS is typable if $f$ is of type
  $\tau(\texttt{arr}\,a'\,b')$,
  $\tau(\texttt{arr}\,a'\,b')\simeq\tau(\texttt{arr}\,a\,b)$, and $x$
  is of type $\tau a$. Then, in this case, the LHS is of type $\tau b$.

  Here, one could be tempted to replace $a'$ by $a$, and $b'$ by $b$,
  so that these conditions are satisfied but this would make the
  rewriting rule non left-linear and the proof of its confluence
  problematic \cite{klop80phd}.

  Fortunately, this is not necessary. Indeed, we can prove that the
  RHS is typable and has the same type as the LHS by using the fact
  that $\tau(\texttt{arr}\,a'\,b')\simeq\tau(\texttt{arr}\,a\,b)$ when
  the LHS is typable. Indeed, in this case, and thanks to the rule
  defining $\tau$, $f$ is of type $\tau a\rightarrow\tau
  b$. Therefore, the RHS has type $\tau b$ as well.
\end{example}

In this paper, we present a new method for doing this kind of
reasoning automatically. By using Knuth-Bendix completion
\cite{knuth67,siekmann83book}, the equations holding when a LHS is
typable are turned into a convergent (\ie confluent and terminating)
set of rewriting rules, so that the type-checking algorithm of Dedukti
itself can be used to check the type of a RHS modulo these equations.

\bigskip
{\bf Outline.} The paper is organized as follows. In Section
\ref{sec-lp}, we recall the definition of the $\lambda\Pi$-calculus
modulo rewriting. In Section \ref{sec-sr}, we recall what it means for
a rewriting rule to preserve typing. In Section \ref{sec-new-algo}, we
describe a new algorithm for checking that a rewriting rule preserves
typing and provide general conditions for ensuring its
termination. Finally, in Section \ref{sec-conclu}, we compare this new
approach with previous ones and conclude. \hide{A prototype implementation by
  Jui-Hsuan Wu is available
  on:\\\url{https://github.com/wujuihsuan2016/lambdapi/tree/sr}.}

%%%%%%%%%%%%%%%%%%%%%%%%%%%%%%%%%%%%%%%%%%%%%%%%%%%%%%%%%%%%%%%%%%%%%%%%%%%%%%
\section{$\lambda\Pi$-calculus modulo rewriting}
\label{sec-lp}

Following Barendregt's book on typed
$\lambda$-calculus \cite{barendregt92chapter}, the
$\lambda\Pi$-calculus is a Pure Type System (PTS) on the set of sorts
$\cS=\{\star,\B\}$:\footnote{PTS sorts should not be confused with the
  notion of sort used in first-order logic. The meaning of these sorts
  will be explained after the definition of typing (Definition
  \ref{dfn-typing}). Roughly speaking, $\star$ is the type of objects
  and proofs, and $\Box$ is the type of set families and predicates.}

\begin{definition}[$\lambda\Pi$-term algebra]
  A $\lambda\Pi$-term algebra is defined by:
  \begin{itemize}
  \item a set $\cF$ of function symbols,
  \item an infinite set $\cV$ of variables,
  \end{itemize}
  \noindent
  such that $\cV$, $\cF$ and $\cS$ are pairwise disjoint.

  The set $\cT(\cF,\cV)$ of $\lambda\Pi$-terms is then inductively defined
  as follows:

$$t,u := s\in\cS\mid x\in\cV\mid f\in\cF\mid \lambda x:t,u\mid tu\mid \Pi x:t,u$$

\noindent
where $\lambda x:t,u$ is called an abstraction, $tu$ an application, and
$\Pi x:t,u$ a (dependent) product (simply written $t\rightarrow u$ if $x$ does
not occur in $u$). As usual, terms are identified modulo renaming of
bound variables ($x$ is bound in $\lambda x:t,u$ and $\Pi x:t,u$).
We denote by $\FV(t)$ the free variables of $t$. A term is said to be closed if
it has no free variables.

A substitution is a finite map from $\cV$ to $\cT(\cF,\cV)$. It is
written as a finite set of pairs. For instance, $\{(x,a)\}$ is the
substitution mapping $x$ to $a$.

Given a substitution $\s$ and a term $t$, we denote by $t\s$ the
capture-avoiding replacement of every free occurrence of $x$ in $t$ by
its image in $\s$.
\end{definition}

\begin{definition}[$\lambda\Pi$-calculus]
  A $\lambda\Pi$-calculus on $\cT(\cF,\cV)$ is given by:
  \begin{itemize}
  \item a function $\T:\cF\rightarrow\cT(\cF,\cV)$ mapping every function
    symbol $f$ to a term $\T_f$ called its type (we will often
    write $f:A$ instead of $\T_f=A$),
  \item a function $\Sigma:\cF\rightarrow\cS$ mapping every function symbol $f$ to a
    sort $\Sigma_f$,
  \item a set $\cR$ of rewriting rules $(l,r)\in\cT^2$,
    written $l\hookrightarrow r$, such that $\FV(r)\sle\FV(l)$.
  \end{itemize}

  We then denote by $\simeq$ the smallest equivalence relation
  containing ${\hookrightarrow}={{\hookrightarrow_\cR}\cup{\hookrightarrow_\beta}}$ where $\hookrightarrow_\cR$ is the smallest
  relation stable by context and substitution containing $\cR$,
  and $\hookrightarrow_\beta$ is the usual $\beta$-reduction relation.
\end{definition}

\begin{example}
  For representing natural numbers, we can use the function symbols
  $N:\star$ of sort $\Box$, and the function symbols $0:N$ and $s:N\rightarrow N$
  of sort $\star$. Addition can be represented by $+:N\rightarrow N\rightarrow N$ of
  sort $\star$ together with the following set of rules:
  $$\begin{array}{r@{~~\hookrightarrow~~}l}
    0+y & y\\
    x+0 & x\\
    x+(sy) & s(x+y)\\
    (sx)+y & s(x+y)\\
    (x+y)+z & x+(y+z)\\
  \end{array}$$
  Note that Dedukti allows overlapping LHS and matching on defined
  symbols like in this example. (It also allows higher-order pattern
  matching like in Combinatory Reduction Systems (CRS)
  \cite{klop93tcs} but we do not consider this feature in the current
  paper.)
\end{example}

Throughout the paper, we assume a given $\lambda\Pi$-calculus
$\Lambda=(\cF,\cV,\Theta,\Sigma,\cR)$.

\begin{definition}[Well-typed terms]\label{dfn-typing}
A typing environment is a possibly empty ordered sequence of pairs
$(x_1,A_1)$, \ldots, $(x_n,A_n)$, written $x_1:A_1,\ldots,x_n:A_n$,
where the $x_i$'s are distinct variables and the $A_i$'s are terms.

A term $t$ has type $A$ in a typing environment $\G$ if the judgment
$\G\vdash t:A$ is derivable from the rules of Figure \ref{fig-lp}. An
environment $\G$ is valid if some term is typable in it.

A substitution $\s$ is a well-typed substitution from an environment
$\G$ to an environment $\G'$, written $\G'\vdash\s:\G$, if, for all
$x:A\in\G$, we have $\G'\vdash x\s:A\s$.
\end{definition}

\begin{figure}
  \caption{Typing rules of the $\lambda\Pi$-calculus modulo rewriting\label{fig-lp}}

  \begin{center}
    \begin{tabular}{ccl}
      (ax) & $\vdash\star:\Box$\\[1mm]
      (fun) & $\cfrac{\vdash\Theta_f:\Sigma_f}{\vdash f:\Theta_f}$\\[3mm]
      (var) & $\cfrac{\G\vdash A:s}{\G,x:A\vdash x:A}$ & ($x\notin\G$)\\[3mm]
      (weak) & $\cfrac{\G\vdash t:T\quad \G\vdash A:s}{\G,x:A\vdash t:T}$ & ($x\notin\G$)\\[3mm]
      (prod) & $\cfrac{\G\vdash A:\star\quad \G,x:A\vdash B:s}{\G\vdash\Pi x:A,B:s}$\\[3mm]
      (app) & $\cfrac{\G\vdash t:\Pi x:A,B\quad \G\vdash a:A}{\G\vdash ta:B\{(x,a)\}}$\\[3mm]
      (abs) & $\cfrac{\G,x:A\vdash b:B\quad \G\vdash\Pi x:A,B:s}{\G\vdash\lambda x:A,b:\Pi x:A,B}$\\[3mm]

      (conv) & $\cfrac{\G\vdash t:T\quad \G\vdash U:s}{\G\vdash t:U}$ & ($T\simeq U$)\\
    \end{tabular}
  \end{center}
\end{figure}

Note that well-typed substitutions preserve typing: if
$\G\vdash t:T$ and $\G'\vdash\s:\G$, then $\G'\vdash t\s:T\s$ \cite{blanqui01phd}.

A type-checking algorithm for the $\lambda\Pi$-calculus modulo
(user-defined) rewriting rules is implemented in the Dedukti tool
\cite{dedukti}.

We first recall a number of basic properties that hold whatever $\cR$
is and can be easily proved by induction on $\vdash$ \cite{blanqui01phd}:

\begin{lemma}\label{lem-basic}
  \begin{enumerate}[(a)]\itemsep=0mm
  \item If $t$ is typable, then every subterm of $t$ is typable.
  \item $\Box$ is not typable.
  \item If $\G\vdash t:T$ then either $T=\Box$ or $\G\vdash T:s$ for some sort $s$.
  \item If $\G\vdash t:\Box$ then $t$ is a kind, that is, of the form $\Pi
    x_1:T_1,\ldots,\Pi x_n:T_n,\star$.
  \item If $\G\vdash t:T$, $\G\sle\G'$ and $\G'$ is valid, then $\G'\vdash t:T$.
  \end{enumerate}
\end{lemma}

Throughout the paper, we assume that, \hide{asm-tf}for all $f$,
$\vdash\Theta_f:\Sigma_f$. Indeed, if $\vdash\Theta_f:\Sigma_f$ does
not hold, then no well-typed term can contain $f$. (This assumption is
implicit in the presentations of LF using signatures
\cite{harper93jacm}.)

More importantly, we will assume that \hide{asm-cr}$\hookrightarrow$ is
confluent on the set $\cT(\cF,\cV)$ of untyped terms, that is, for all
terms $t,u,v\in\cT(\cF,\cV)$, if $t\hookrightarrow^*u$ and
$t\hookrightarrow^*v$, then there exists a term $w\in\cT(\cF,\cV)$ such
that $u\hookrightarrow^*w$ and $v\hookrightarrow^*w$, where $\hookrightarrow^*$ is
the reflexive and transitive closure of $\hookrightarrow$.

This condition is required for ensuring that conversion behaves
well with respect to products (if $\Pi x:A,B\simeq \Pi x:A',B'$ then
$A\simeq A'$ and $B\simeq B'$), which in particular implies
subject-reduction for $\hookrightarrow_\beta$.

This last assumption may look strong, all the more so since confluence
is undecidable. However this property is satisfied by many systems
in practice. For instance, $\hookrightarrow$ is confluent if the
left-hand sides of $\cR$ are algebraic (Definition \ref{dfn-alg}),
linear and do not overlap with each other \cite{oostrom94phd}. This is
in particular the case of the rewriting systems corresponding to the
function definitions allowed in functional programming languages such
as Haskell, Agda, OCaml or Coq. But confluence can be relaxed in some
cases: when there are no type-level rewriting rules
\cite{barbanera97jfp} or when the right-hand sides of type-level
rewriting rules are not products \cite{blanqui05mscs}.

When $\hookrightarrow$ is confluent, the typing relation satisfies additional
properties. For instance, the set of typable terms can be divided into
three disjoint classes:
\begin{itemize}\itemsep=0mm
\item the terms of type $\Box$, called kinds, of the form $\Pi
  x_1:A_1,\ldots,\Pi x_n:A_n,\star$;
\item the terms whose type is a kind, called predicates;
\item the terms whose type is a predicate, called objects.
\end{itemize}

%%%%%%%%%%%%%%%%%%%%%%%%%%%%%%%%%%%%%%%%%%%%%%%%%%%%%%%%%%%%%%%%%%%%%%%%%%%%%%
\section{Subject-reduction}
\label{sec-sr}

A relation $\tgt$ preserves typing (subject-reduction property) if,
for all environments $\G$ and all terms $t$, $u$ and $A$, if $\G\vdash
t:A$ and $t\tgt u$, then $\G\vdash u:A$.

One can easily check that $\hookrightarrow_\beta$ preserves typing when $\hookrightarrow$ is
confluent \cite{blanqui01phd}. Our aim is therefore to check that
$\hookrightarrow_\cR$ preserves typing too. To this end, it is enough to check that
every rule $l\hookrightarrow r\in\cR$ preserves typing, that is, for all
environments $\G$, substitutions $\s$ and terms $A$, if $\G\vdash l\s:A$,
then $\G\vdash r\s:A$.

A first idea is to require that:

\begin{center}
  (*) there exist $\D$ and $B$ such that $\D\vdash l:B$ and $\D\vdash r:B$.
\end{center}

\noindent But this condition is not sufficient in general as shown by
the following example:

\begin{example}
  Consider the rule $f(xy)\hookrightarrow y$ with $f:B\rightarrow B$. In the environment
  $\D={x:B\rightarrow B,y:B}$, we have $\D\vdash l:B$ and $\D\vdash r:B$. However, in
  the environment $\G={x:A\rightarrow B,y:A}$, we have $\G\vdash l:B$ and $\G\vdash
  r:A$.
\end{example}

The condition (*) is sufficient if the rule left-hand side is a
non-variable simply-typed first-order term \cite{barbanera97jfp}, a
notion that we slightly generalize as follows:

\begin{definition}[Pattern]\label{dfn-alg}
  We assume that the set of variables is split in two disjoint sets,
  the algebraic variables and the non-algebraic ones, and that there
  is an injection $\h~$ from algebraic variables to non-algebraic
  variables.

  A term is algebraic if it is an algebraic variable or of the form
  $ft_1\ldots t_n$ with each $t_i$ algebraic and $f$ a function symbol
  whose type is of the form $\Pi x_1:A_1,\ldots,x_n:A_n,B$.

  A term is an object-level algebraic term if it is algebraic and all
  its function symbols are of sort $\star$.

  A pattern is an algebraic term of the form $ft_1\ldots t_n$ where
  each $t_i$ is an object-level algebraic term.
\end{definition}

The distinction between algebraic and non-algebraic variables is
purely technical: for generating equations (Definition
\ref{dfn-typ-constr}), we need to associate a type $\h{x}$ to every
variable $x$, and we need those variables $\h{x}$ to be distinct from
one another and distinct from the variables used in rules. To do so,
we split the set of variables into two disjoint sets. The ones used in
rules are called algebraic, and the others are called
non-algebraic. Finally, we ask the function $\h~$ to be an injection
from the set of algebraic variables to the set of non-algebraic
variables.

In the rest of the paper, we also assume that \hide{asm-lhs-alg} rule
left-hand sides are patterns. Hence, every rule is of the form
$f\,l_1\ldots l_n\hookrightarrow r$, and we say that a symbol $f\in\cF$ is
defined if there is in $\cR$ a rule of the form $f\,l_1\ldots
l_n\hookrightarrow r$.

However, the condition (*) is not satisfactory in the context of
dependent types. Indeed, when function symbols have dependent types,
it often happens that a term is typable only if it is non-linear. And,
with non-left-linear rewriting rules, $\hookrightarrow$ is generally not
confluent on untyped terms \cite{klop80phd}, while there exist many
confluence criteria for left-linear rewriting systems
\cite{oostrom94phd}.

Throughout the paper, we will use the following simple but
paradigmatic example to illustrate how our new algorithm works:

\begin{example}
  Consider the following rule to define the $\mi{tail}$ function on
  vectors:

  $$\mi{tail}~n~(\mi{cons}~x~p~v)\hookrightarrow v$$

  \noindent
  where $\mi{tail}:{\Pi n:N,V(sn)\rightarrow Vn}$, $V:{N\rightarrow\star}$,
  $\mi{nil}:{V0}$, $\mi{cons}:{R\rightarrow\Pi n:N,Vn\rightarrow}$ ${V(sn)}$ and
  $R:\star$.

  For the left-hand side to be typable, we need to take $p=n$,
  because $\mi{tail}~n$ expects an argument of type $V(sn)$, but
  $\mi{cons}~x~p~v$ is of type $V(sp)$.

  Yet, the rule with $p\neq n$ preserves typing. Indeed, assume that
  there is an environment $\G$, a substitution $\s$ and a term $A$
  such that $\G\vdash\mi{tail}~n\s~(\mi{cons}~x\s~p\s~v\s):A$. By
  inversion of typing rules, we get $V(n\s)\simeq A$, $\G\vdash A:s$ for
  some sort $s$, $V(sp\s)\simeq V(sn\s)$ and $\G\vdash v\s:Vp\s$. Assume
  now that $V$ and $s$ are undefined, that is, there is no rule of
  $\cR$ of the form $Vt\hookrightarrow u$ or $st\hookrightarrow u$. Then, by confluence,
  $p\s\simeq n\s$. Therefore, $Vp\s\simeq A$ and $\G\vdash v\s:A$.
\end{example}

Hence, that a rewriting rule $l\hookrightarrow r$ preserves typing does not mean
that its left-hand side $l$ must be typable
\cite{blanqui05mscs}. Actually, if no instance of $l$ is typable, then
$l\hookrightarrow r$ trivially preserves typing (since it can never be applied)!
The point is therefore to check that any typable instance of $l\hookrightarrow r$
preserves typing.

%%%%%%%%%%%%%%%%%%%%%%%%%%%%%%%%%%%%%%%%%%%%%%%%%%%%%%%%%%%%%%%%%%%%%%%%%%%%%%
\section{A new subject-reduction criterion}
\label{sec-new-algo}

The new criterion that we propose for checking that $l\hookrightarrow r$
preserves typing proceeds in two steps. First, we generate conversion
constraints that are satisfied by every typable instance of $l$
(Figure \ref{fig-typ-constr}). Then, we try to check that $r$ has the
same type as $l$ in the type system where the conversion relation is
extended with the equational theory generated by the conversion
constraints inferred in the first step. For type-checking in this
extended type theory to be decidable and implementable using Dedukti
itself, we use Knuth-Bendix completion \cite{knuth67} to replace the
set of conversion constraints by an equivalent but convergent (\ie
terminating and confluent) set of rewriting rules.

%%%%%%%%%%%%%%%%%%%%%%%%%%%%%%%%%%%%%%%%%%%%%%%%%%%%%%%%%%%%%%%%%%%%%%%%%%%%%%
\subsection{Inference of typability constraints}

We first define an algorithm for inferring typability constraints and
then prove its correctness and completeness.

\begin{definition}[Typability constraints]\label{dfn-typ-constr}
  For every algebraic term $t$, we assume given a valid environment
  $\D_t={{\h{y_1}:\star},{y_1:\h{y_1}},\ldots,{\h{y_k}:\star},{y_k:\h{y_k}}}$ where
  $y_1,\ldots,y_k$ are the free variables of $t$.
  
  Let $\uparrow$ be the partial function defined in Figure
  \ref{fig-typ-constr}. It takes as input a term $t$ and returns a
  pair $(A,\cE)$, written $A[\cE]$, where $A$ is a term and $\cE$ is a
  set of equations, an equation being a pair of terms $(l,r)$ usually
  written $l=r$.

  A substitution $\s$ satisfies a set $\cE$ of equations, written
  $\s\models\cE$, if for all equations ${a=b}\in{\cE}$, $a\s\simeq b\s$.
\end{definition}

\begin{figure}[ht]
  \caption{Typability constraints\label{fig-typ-constr}}
  $$\cfrac{}{y\uparrow\h{y}[\emptyset]}$$

  $$\cfrac{f:\Pi x_1:T_1,\ldots,\Pi x_n:T_n,U\quad t_1\uparrow A_1[\cE_1]\quad t_n\uparrow A_n[\cE_n]}{\begin{array}{c}ft_1\ldots t_n\uparrow U\s[\cE_1\cup\ldots\cup\cE_n\cup\{A_1=T_1\s,\ldots,A_n=T_n\s\}]\\\text{where }\s=\{(x_1,t_1),\ldots,(x_n,t_n)\}\end{array}}$$
\end{figure}

\begin{example}
  In our running example $\mi{tail}~n~(\mi{cons}~x~p~v)\hookrightarrow v$,
  we have $\mi{cons}~x~p~v\uparrow V(sp)[\cE_1]$ with
  $\cE_1=\{{\h{x}=T},{\h{p}=N},{\h{v}=Vp}\}$, and
  $\mi{tail}~n~(\mi{cons}~x~p~v)\uparrow Vn[\cE_2]$ with
  $\cE_2=\cE_1\cup\{{\h{n}=N},{V(sp)=V(sn)}\}$.
\end{example}

\begin{lemma}\label{lem-prod}
  If $\G\vdash\Pi x_1:T_1,\ldots,\Pi x_n:T_n,U:s$ then, for all $i$,
  $\G^{i-1}\vdash T_i:\star$ and $\G^n\vdash U:s$, where
  $\G^i=\G,x_1:T_1,\ldots,x_i:T_i$.
\end{lemma}

\begin{proof}
Since $\hookrightarrow$ is confluent and left-hand sides are patterns, $s\simeq s'$
iff $s=s'$. The result follows then by inversion of typing rules and
weakening.\\
\end{proof}

In particular, because \hide{asm-tf}$\vdash\Theta_f:\Sigma_f$ for all
$f$, we have:

\begin{corollary}\label{cor-fun-type}
  For all function symbols $f:\Pi x_1:T_1,\ldots,\Pi x_n:T_n,U$ and
  integer $i$, we have $\G_f^{i-1}\vdash T_i:\star$ and $\G_f^n\vdash
  U:\Sigma_f$ where $\G_f^i=x_1:T_1,\ldots,x_i:T_i$.
\end{corollary}

\begin{lemma}\label{lem-app}
  For all environments $\G$, terms $t,x_1,T_1,\ldots,x_n,T_n,U,T$ and
  substitutions $\s$ for $x_1,\ldots,x_n$, if $\G\vdash t:\Pi
  x_1:T_1,\ldots,\Pi x_n:T_n,U$ and $\G\vdash tx_1\s\ldots x_n\s:T$,
  then $U\s\simeq T$ and $\G\vdash\s:\D^n$ where
  $\D^n=x_1:T_1,\ldots,x_n:T_n$.
\end{lemma}

\begin{proof}
  Let $\s_i=\{(x_1,t_1),\ldots,(x_{i-1},t_{i-1})\}$. We proceed by
  induction on $n$.
  \begin{itemize}
  \item Case $n=0$. By equivalence of types.
  \item Case $n>0$. By inversion of typing rules and weakening,
    $\G,\D^{n-1}\vdash T_n:\star$, $\G\vdash tx_1\s_{n-1}\ldots
    x_{n-1}\s_{n-1}:\Pi x_n:A,B$, $\G\vdash x_n\s:A$ and
    $B\{(x_n,x_n\s)\}\simeq T$. By induction hypothesis,
    $\G\vdash\s_{n-1}:\D^{n-1}$ and $(x_n:T_n\s_{n-1})U\s_{n-1}\simeq\Pi
    x_n:A,B$. By substitution, $\G\vdash T_n\s_{n-1}:\star$. By
    confluence\hide{asm-cr}, $T_n\s_{n-1}\simeq A$ and
    $U\s_{n-1}\simeq B$. Therefore, by conversion, $\G\vdash x_n\s:T_n\s$
    and $\G\vdash\s:\D^n$. Now, $x_n$ can always be chosen so that
    $U\s=U\s_{n-1}\{(x_n,x_n\s)\}$. Therefore, $U\s\simeq
    B\{(x_n,x_n\s)\}\simeq T$.
  \end{itemize}
\end{proof}

\begin{lemma}\label{lem-infer}
  \begin{itemize}
  \item (Correctness) For all algebraic terms $t$, terms $T$ and sets
    of equations $\cE$, if $t\uparrow T[\cE]$ then, for all valid environments
    $\G$, substitutions $\h\theta$ such that $\G\vdash\h\theta:\D_t$ and
    $\h\theta\models\cE$, we have $\G\vdash t\h\theta:T\h\theta$.

  \item (Completeness) For all environments $\G$, patterns $t$,
    substitutions $\theta$ and terms $A$, if $\G\vdash t\theta:A$, then there are
    a term $T$, a set of equations $\cE$ and a substitution $\h\theta$
    extending $\theta$ such that $t\uparrow T[\cE]$, $\h\theta\models\cE$,
    $\G\vdash\h\theta:\D_t$ and $A\simeq T\h\theta$.
  \end{itemize}
\end{lemma}

\begin{proof}
  \begin{itemize}
  \item (Correctness) By induction on $t$.
    \begin{itemize}
    \item Case $t=y$. Then, $T=\h{y}$ and $\cE=\emptyset$. By assumption, we
      have $\G\vdash y\theta:\h{y}\theta$. Therefore, $\G\vdash t\theta:T\theta$.
      
    \item Case $t=ft_1\ldots t_n$ with $f:\Pi
      x_1:T_1,\ldots,x_n:T_n,U$, $t_1\uparrow A_1[\cE_1]$, \ldots, $t_n\uparrow
      A_n[\cE_n]$. Then, $T=U\s$ and $\cE=\cE_1\cup\ldots\cup
     \cE_n\cup\{A_1=T_1\s,\ldots,A_n=T_n\s\}$ where
      $\s=\{(x_1,t_1),\ldots,(x_n,t_n)\}$.

      By Lemma \ref{lem-prod}, we have $\G_f^{i-1}\vdash T_i:\star$.

      By induction hypothesis, for all $i$, we have $\G\vdash
      t_i\h\theta:A_i\h\theta$ and $A_i\h\theta\simeq T_i\s\h\theta$.

      We now prove that, for all $i$, $\G\vdash T_i\s\h\theta:\star$ and
      $\G\vdash x_i\s\h\theta:T_i\s\h\theta$, hence that $\G\vdash\s:\G_f^i$, by
      induction on $i$.
      \begin{itemize}
      \item Case $i=1$. Since $\vdash T_1:\star$, $T_1$ is closed and
        $T_1\s\h\theta=T_1$. Therefore, by weakening, $\G\vdash
        T_1\s\h\theta:\star$ and, by conversion, $\G\vdash
        x_1\s\h\theta:T_1\s\h\theta$.
        
      \item Case $i>1$. By induction hypothesis,
        $\G\vdash\s\h\theta:\G_f^{i-1}$. Since $\G_f^{i-1}\vdash T_i:\star$, by
        substitution, we get $\G\vdash T_i\s\h\theta:\star$. Therefore, by
        conversion, $\G\vdash x_i\s\h\theta:T_i\s\h\theta$.
      \end{itemize}

      Hence, $\G\vdash\s\h\theta:\G_f^n$. Now, since $\G_f^n\vdash fx_1\ldots
      x_n:U$, by substitution, we get $\G\vdash t:U\s\h\theta$.
    \end{itemize}

  \item (Completeness) We first prove completeness for object-level
    algebraic terms $t$ such that $\G\vdash A:\star$, by induction on
    $t$.
    \begin{itemize}
    \item Case $t=y$. We take $T=\h{y}$, $\cE=\emptyset$ and
      $\h\theta=\theta\cup\{(\h{y},A)\}$. We have $t\uparrow T[\cE]$, $\h\theta\models\cE$
      and $A\simeq T\h\theta$. Now, $\G\vdash y\h\theta:\h{y}\h\theta$ and $\G\vdash
      \h{y}\h\theta:\star$. Therefore, $\G\vdash\h\theta:\D_t$.

    \item Case $t=ft_1\ldots t_n$ with $f:\Pi
      x_1:T_1,\ldots,x_n:T_n,U$. By Lemma \ref{lem-prod}, for all $i$,
      we have $\G_f\vdash x_i:T_i$ and $\G_f\vdash T_i:\star$, where
      $\G_f=x_1:T_1,\ldots,x_n:T_n$. By Lemma \ref{lem-app} because
      \hide{asm-tf}$\vdash\Theta_f:\Sigma_f$ for all $f$, we have $A\simeq
      U\s\theta$ and $\G\vdash\s\theta:\G_f$. Hence, by substitution,
      for all $i$, we have $\G\vdash t_i\theta:T_i\s\theta$ and
      $\G\vdash T_i\s\theta:\star$. Therefore, by induction
      hypothesis, there are $A_i$, $\cE_i$ and $\h\theta_i$ extending
      $\theta$ such that $t_i\uparrow A_i[\cE_i]$,
      $\h\theta_i\models\cE_i$, $\G\vdash\h\theta_i:\D_{t_i}$ and
      $T_i\s\theta\simeq A_i\h\theta_i$. Then, let $T=U\s$,
      ${\cE}={{\cE_1}\cup\ldots\cup{\cE_n}\cup{\{(A_1,T_1\s),\ldots,(A_n,T_n\s)\}}}$,
      $y\h\theta=y\theta$ if $y\in\FV(t)$, and
      $\h{y}\h\theta=\h{y}\h\theta_i$ where $i$ is the smallest
      integer such that $y\in\FV(t_i)$. Then, we have $t\uparrow
      T[\cE]$ and $A\simeq U\s\theta=T\h\theta$.

      If $y\in\FV(t_i)\cap\FV(t_j)$, then $y\h\theta_i=y\theta=y\h\theta_j$ since
      $\h\theta_i$ and $\h\theta_j$ are both extensions of $\theta$. Now, if
      $\G\vdash y\h\theta_i:\h{y}\h\theta_i$ and $\G\vdash y\h\theta_j:\h{y}\h\theta_j$
      then, by equivalence of types,
      $\h{y}\h\theta_i\simeq\h{y}\h\theta_j$. Therefore, $\h\theta\models\cE$ and
      $\G\vdash\h\theta:\D_t$.
    \end{itemize}

    Let now $t$ be a pattern. By definition, $t$ is of the form
    $ft_1\ldots t_n$ with $f:\Pi x_1:T_1,\ldots,x_n:T_n,U$ and each
    $t_i$ an object-level algebraic term. As we have seen above, for
    all $i$, we have $\G\vdash t_i\theta:T_i\s\theta$ and $\G\vdash
    T_i\s\theta:\star$. Therefore, by completeness for object level
    algebraic terms, there are $A_i$, $\cE_i$ and $\h\theta_i$ extending
    $\theta$ such that $t_i\uparrow A_i[\cE_i]$, $\h\theta_i\models\cE_i$,
    $\G\vdash\h\theta_i:\D_{t_i}$ and $T_i\s\theta\simeq A_i\h\theta_i$. We can now
    conclude like in the previous case.
  \end{itemize}
\end{proof}

\begin{example}
  In our running example $\mi{tail}~n~(\mi{cons}~x~p~v)\hookrightarrow v$,
  we have seen that $\mi{cons}~x~p~v\uparrow V(sp)[\cE_1]$ with
  $\cE_1=\{{\h{x}=T},{\h{p}=N}$, ${\h{v}=Vp}\}$, and
  $\mi{tail}~n~(\mi{cons}~x~p~v)$ $\uparrow Vn[\cE_2]$ with
  $\cE_2=\cE_1\cup\{{\h{n}=N},{V(sp)=V(sn)}\}$. This means that, if
  $\s$ is a substitution and $(\mi{tail}~n~(\mi{cons}~x~p~v))\s$ is
  typable, then $\s\models\cE_2$. In particular, $V(sp\s)\simeq
  V(sn\s)$.
\end{example}

%%%%%%%%%%%%%%%%%%%%%%%%%%%%%%%%%%%%%%%%%%%%%%%%%%%%%%%%%%%%%%%%%%%%%%%%%%%%%%
\subsection{Type-checking modulo typability constraints}

For checking that the right-hand side of a rewriting rule $l\hookrightarrow r$ has
the same type as the left-hand side modulo the typability constraints
$\cE$ of the left hand-side, we introduce a new $\lambda\Pi$-calculus as
follows:

\begin{definition}
  Given a pattern $l$ and a set of equations $\cE$ such that $\cR$
  contains no variable of $\{x\mid x\in \FV(l)\}\cup\{\h{x}\mid
  x\in\FV(l)\}$\footnote{This can always be done by renaming
    variables.}, we define a new $\lambda\Pi$-calculus
  $\Lambda_{l,\cE}=(\cF',\cV',\Theta',\Sigma',\cR')$ where:
\begin{itemize}
\item $\cF'=\cF\cup\{x\mid x\in\FV(l)\}\cup\{\h{x}\mid x\in\FV(l)\}$
\item $\cV'=\cV-(\{x\mid x\in\FV(l)\}\cup\{\h{x}\mid x\in\FV(l)\})$
\item $\Theta'=\Theta\cup\{(x,\h{x})\mid x\in\FV(l)\}\cup\{(\h{x},\star)\mid
  x\in\FV(l)\}$
\item $\Sigma'=\Sigma\cup\{(x,\star)\mid x\in\FV(l)\}\cup\{(\h{x},\B)\mid x\in\FV(l)\}$
\item $\cR'=\cR\cup\cE\cup\cE^{-1}$, where $l=r\in\cE^{-1}$ iff $r=l\in\cE$.
\end{itemize}
We denote by $\simeq_{l,\cE}$ the conversion relation of $\Lambda_{l,\cE}$,
and by $\vdash_{l,\cE}$ its typing relation.
\end{definition}

$\Lambda_{l,\cE}$ is similar to $\Lambda$ except that the symbols of $\{x\mid
x\in\FV(l)\}\cup\{\h{x}\mid x\in\FV(l)\}$ are not variables but
function symbols, and that the set of rewriting rules is extended by
$\cE\cup\cE^{-1}$ which, in $\Lambda_{l,\cE}$, is a set of closed rewriting
rules (rules and equations are synonyms: they both are pairs of
terms).

\begin{lemma}\label{lem-simeq}
  For all patterns $l$, sets of equations $\cE$ and substitutions $\s$
  in $\Lambda$, and for all terms $t,u$ in $\Lambda_{l,\cE}$, if $\s\models\cE$
  and $t\simeq_{l,\cE}u$, then $t\s\simeq u\s$.\footnote{Note that,
    here, we extend the notion of substitution by taking maps on
    $\cV\cup\cF$.}
\end{lemma}

\begin{proof}
  Immediate as each application of an equation
  $(a,b)\in\cE\cup\cE^{-1}$ can be replaced by a conversion $a\simeq
  b$.
\end{proof}

\begin{theorem}\label{thm-sr}
  For all patterns $l$, sets of equations $\cE$, and terms $T,r$ in
  $\Lambda$, if $l\uparrow T[\cE]$ and $\vdash_{l,\cE} r:T$, then $l\hookrightarrow r$ preserves
  typing in $\Lambda$.
\end{theorem}

\begin{proof}
Let $\D$ be an environment, $\s$ be a substitution and $A$ be a term
of $\Lambda$ such that $\D\vdash l\s:A$. By Lemma \ref{lem-infer}
(completeness), there are a term $T'$, a set of equations $\cE'$ and a
substitution $\h\s$ extending $\s$ such that $t\uparrow T'[\cE']$,
$\h\s\models\cE'$, $\D\vdash\h\s:\D_t$ and $A\simeq T'\h\s$. Since $\uparrow$
is a function, we have $T'=T$ and $\cE'=\cE$.

We now prove that, if $\G\vdash_{l,\cE} t:T$, then $\D,\G\h\s\vdash
t\h\s:T\h\s$, by induction on $\vdash_{l,\cE}$ (note that $\h\s$ replaces
function symbols by terms).
\begin{itemize}
\item[(fun)] $\cfrac{\vdash_{l,\cE}\Theta'_f:\Sigma'_f}{\vdash_{l,\cE}f:\Theta'_f}$. By induction
  hypothesis, we have $\D\vdash\Theta'_f\h\s:\Sigma'_f\h\s=\Sigma'_f$.
  \begin{itemize}
  \item Case $f\in\cF$. Then, $f\h\s=f$, $\Theta'_f\h\s=\Theta'_f=\Theta_f$ and
    $\Sigma'_f=\Sigma_f$. By inverting typing rules, we get
    $\vdash\Theta_f:\Sigma_f$. Therefore, by (fun) and (weak), $\D\vdash f:\Theta_f$,
    that is, $\D\vdash f\h\s:\Theta_f\h\s$.
  \item Case $f=x\in\FV(l)$. Then, $f\h\s=x\s$ and
    $\Theta'_f\h\s=\h{x}\h\s$. Therefore, $\D\vdash f\h\s:\Theta_f\h\s$ since
    $\D\vdash x\h\s:\h{x}\h\s$.
  \item Case $f=\h{x}$ with $x\in\FV(l)$. Then, $f\h\s=\h{x}\h\s$ and
    $\Theta'_f\h\s=\star$. Therefore, $\D\vdash f\h\s:\Theta_f\h\s$ since
    $\D\vdash\h{x}\h\s:\star$.
  \end{itemize}

\item[(conv)] $\cfrac{\G\vdash_{l,\cE} t:T\quad T\simeq_{l,\cE} U\quad
  \G\vdash_{l,\cE} U:s}{\G\vdash_{l,\cE} t:U}$. By induction hypothesis,
  $\D,\G\h\s\vdash t\h\s:T\h\s$ and $\D,\G\h\s\vdash U\h\s:s$. By Lemma
  \ref{lem-simeq}, $T\h\s\simeq U\h\s$ since
  $\h\s\models\cE$. Hence, by (conv), $\D,\G\h\s\vdash t\h\s:U\h\s$.
  
\item The other cases follow easily by induction hypothesis.
\end{itemize}

Hence, we have $\D\vdash r\h\s:T\h\s$. Since $\FV(r)\sle\FV(l)$, we have
$r\h\s=r\s$. Since $\D\vdash l\s:A$, by Lemma \ref{lem-basic}, either
$\D\vdash A:s$ for some sort $s$, or $A=\Box$ and $l\s$ is of the form
$\Pi x_1:A_1,\ldots,\Pi x_k:A_k,\star$. Since $l$ is a pattern, $l$
is of the form $fl_1\ldots l_n$. Therefore, $\D\vdash A:s$ and, by
(conv), $\D\vdash r\s:A$.
\end{proof}

\begin{example}
  We have seen that $\mi{cons}~x~p~v\uparrow V(sp)[\cE_1]$ with
  $\cE_1=\{{\h{x}=T},{\h{p}=N}$, ${\h{v}=Vp}\}$, and
  $\mi{tail}~n~(\mi{cons}~x~p~v)\uparrow Vn[\cE_2]$ with
  $\cE_2=\cE_1\cup\{{\h{n}=N},{V(sp)=V(sn)}\}$. After the previous
  theorem, the rewriting rule defining $\mi{tail}$ preserves typing if
  we can prove that $\vdash_{l,\cE_2}v:Vn$ where, in $\Lambda_{l,\cE_2}$, $v$
  is a function symbol of type $\h{v}$ and sort $\star$, and types are
  identified modulo $\simeq$ and the equations of $\cE_2$. But this is
  not possible since $v:Vp$ and $Vp\not\simeq_{l,\cE_2}Vn$. Yet, if
  $\s\models V(sp)=V(sn)$ and $V$ and $s$ are undefined then, by
  confluence, $\s\models p=n$ and thus $\s\models Vp=Vn$. We therefore
  need to simplify the set of equations before type-checking the
  right-hand side.
\end{example}

%%%%%%%%%%%%%%%%%%%%%%%%%%%%%%%%%%%%%%%%%%%%%%%%%%%%%%%%%%%%%%%%%%%%%%%%%%%%%%
\subsection{Simplification of typability constraints}

In this section, we show that Theorem \ref{thm-sr} can be generalized
by using any valid simplification relation, and give an example of
such a relation.

\begin{definition}[Valid simplification relation]
  A relation $\leadsto$ on sets of equations is valid if, for all sets
  of equations $\cD,\cD'$ and substitutions $\s$, if $\s\models\cD$
  and $\cD\leadsto\cD'$, then $\s\models\cD'$.
\end{definition}

Theorem \ref{thm-sr} can be easily generalized as follows:

\begin{theorem}[Preservation of typing]\label{thm-sr2}
  For all patterns $l$, sets of equations $\cD,\cE$, and terms $T,r$
  in $\Lambda$, if $l\uparrow T[\cD]$, $\cD\leadsto^*\cE$ and
  $\vdash_{l,\cE} r:T$, then $l\hookrightarrow r$ preserves typing in
  $\Lambda$.
\end{theorem}

We have seen in the previous example that, thanks to confluence,
$\sigma\models p=n$ whenever $\sigma\models sp=sn$ and $s$ is
undefined. But this last condition is a particular case of a more
general property:

\begin{definition}[$I$-injectivity]
  Given $f:\Pi x_1:T_1,\ldots,\Pi x_n:T_n,U$ and a set
  $I\sle\{1,\ldots,n\}$, we say that $f$ is $I$-injective when, for
  all $t_1,u_1,\ldots,t_n,u_n$, if $ft_1\ldots t_n\simeq fu_1\ldots
  u_n$ and, for all $i\notin I$, $t_i\simeq u_i$, then, for all $i\in
  I$, $t_i\simeq u_i$.
\end{definition}

For instance, $f$ is $\{1,\ldots,n\}$-injective if $f$ is
undefined. The new version of Dedukti allows users to declare if a
function symbol is $I$-injective (like the function $\tau$ in Example
\ref{ex-beta}), and a procedure for checking $I$-injectivity of
function symbols defined by rewriting rules has been developed and
implemented in Dedukti \cite{wu19tr}. For instance, the function
symbol $\tau$ of Example \ref{ex-beta}, which is defined by the rule
$\tau(\mr{arr}\,x\,y)\hookrightarrow\tau x\rightarrow\tau y$, can be
proved to be $\{1\}$-injective.

Clearly, $I$-injectivity can be used to define a valid simplification
relation. In fact, one can easily check that the following
simplification rules are valid too:

\begin{lemma}
 The relation defined in Figure \ref{fig-simpl} is a valid
 simplification relation.
\end{lemma}

\begin{proof}
  We only detail the first rule which says that, if some substitution
  $\sigma$ validates some equation $t=u$, that is, if $t\sigma\simeq
  u\sigma$, then $\sigma$ validates any equation $t'=u'$ where $t'$
  and $u'$ are reducts of $t$ and $u$ respectively. Indeed, since $t'$
  is a reduct of $t$, $t'\simeq t$. Similarly, $u'\simeq
  u$. Therefore, by stability of conversion by substitution and
  transitivity, $t'\sigma\simeq u'\sigma$.
\end{proof}

\begin{figure}[ht]
  \caption{Some valid simplification rules on typability constraints\label{fig-simpl}}
  $$\begin{array}{rcl}
    {\cD}\uplus{\{t=u\}}
    &\leadsto& {\cD}\cup{\{t'=u'\}}\text{ if $t\hookrightarrow^*t'$ and $u\hookrightarrow^*u'$}\\

    {\cD}\uplus{\{\Pi x:t_1,t_2=\Pi x:u_1,u_2\}}
    &\leadsto& {\cD}\cup{\{t_1=u_1,t_2=u_2\}}\text{ if $x$ is fresh}\\

    {\cD}\uplus{\{ft_1\ldots t_n=fu_1\ldots u_n\}}
    &\leadsto& {\cD}\cup{\{t_i=u_i\mid i\in I\}}\\
    &&\text{if $f$ is $I$-injective and }\all i\notin I, t_i\simeq_{l,\cD}u_i\\
  \end{array}$$
\end{figure}

\begin{example}
  We can now handle our running example. We have $\mi{cons}~x~p~v\uparrow
  V(sp)[\cE_1]$ with $\cE_1=\{{\h{x}=T},{\h{p}=N},{\h{v}=Vp}\}$,
  $l=\mi{tail}~n~(\mi{cons}~x~p~v)\uparrow Vn[\cE_2]$ with
  $\cE_2=\cE_1\cup\{{\h{n}=N},{V(sp)=V(sn)}\}$, and
  $\cE_2\leadsto^*\cE_2'=\cE_1\cup\{{\h{n}=N},{p=n}\}$ since $V$ and
  $s$ are $\{1\}$-injective. Therefore, $\vdash_{l,\cE_2'}v:Vn$ and $l\hookrightarrow
  v$ preserves typing.
\end{example}

The above simplification relation works for the rewriting rule
defining $\mi{tail}$ but may not be sufficient in more general
situations:

\begin{example}
  Let $\cD$ be the set of equations $\{fct=ga,fcu=gb,a=b\}$ and assume
  that $f$ is $\{2\}$-injective. Then the equation $t=u$ holds as
  well, but $\cD$ cannot be simplified by the above rules because it
  contains no equation of the form $fct=fcu$.
\end{example}

We leave for future work the development of more general
simplification relations.

%%%%%%%%%%%%%%%%%%%%%%%%%%%%%%%%%%%%%%%%%%%%%%%%%%%%%%%%%%%%%%%%%%%%%%%%%%%%%%
\subsection{Decidability conditions}

We now discuss the decidability of type-checking in $\Lambda_{l,\cE}$ and
of the simplification relation based on injectivity, assuming that
${\hookrightarrow_\beta}\cup{\hookrightarrow_\cR}$ is terminating and confluent so that type-checking is
decidable in $\Lambda$. In both cases, we have to decide $\simeq_{l,\cE}$,
the reflexive, symmetric and transitive closure of
${\hookrightarrow_\beta}\cup{\hookrightarrow_\cR}\cup{\hookrightarrow_\cE}\cup{\hookrightarrow_{\cE^{-1}}}$, where $\cE$ is a set
of closed equations.

As it is well known, an equational theory is decidable if there exists
a convergent (\ie terminating and confluent) rewriting system having
the same equational theory: to decide whether two terms are
equivalent, it suffices to check that their normal forms are
identical.

In \cite{knuth67}, Knuth and Bendix introduced a procedure to compute
a convergent rewriting system included in some termination ordering,
when equations are algebraic. Interestingly, this procedure always
terminates when equations are closed, if one takes a termination
ordering that is total on closed terms like the lexicographic path
ordering $>_{lpo}$ wrt any total order $>$ on function symbols (for more
details, see for instance \cite{baader98book}).

For the sake of self-contentness, we recall in Figure \ref{fig-kb} a
rule-based definition of closed completion. These rules operate on a
pair $(\cE,\cD)$ made of a set of equations $\cE$ and a set of rules
$\cD$. Starting from $(\cE,\emptyset)$, completion consists in applying
these rules as long as possible. This process necessarily ends on
$(\emptyset,\cD)$ where $\cD$ is terminating (because
${\cD}\sle{>_{lpo}}$) and confluent (because it has no critical
pairs).

\begin{figure}[ht]
  \caption{Rules for closed completion\label{fig-kb}}
  $$\begin{array}{rcl}
    (\cE\uplus\{l=r\},\cD) &\leadsto& (\cE,\{l\hookrightarrow r\}\cup\cD)\text{ if } l>r\\
    (\cE\uplus\{l=r\},\cD) &\leadsto& (\cE,\{r\hookrightarrow l\}\cup\cD)\text{ if } l<r\\
    (\cE\uplus\{t=t\},\cD) &\leadsto& (\cE,\cD)\\
    (\cE,\{l[g]\hookrightarrow r,g\hookrightarrow d\}\uplus\cD) &\leadsto& (\cE\cup\{l[d]=r\},\{g\hookrightarrow d\}\cup\cD)\\
    (\cE,\{l\hookrightarrow r[g],g\hookrightarrow d\}\uplus\cD) &\leadsto& (\cE,\{l\hookrightarrow r[d],g\hookrightarrow d\}\cup\cD)\\
  \end{array}$$
\end{figure}

We leave for future work the extension of this procedure to the case
of non-algebraic, and possibly higher-order, equations.

If we apply this procedure to the set $\cE$ of equations (assuming
that they are algebraic), we get that $\simeq_{l,\cE}$ is the
reflexive, symmetric and transitive closure of
${\hookrightarrow_\beta}\cup{\hookrightarrow_\cR}\cup{\hookrightarrow_\cD}$, where ${\hookrightarrow_\beta}\cup{\hookrightarrow_\cR}$ and $\hookrightarrow_\cD$ are
both terminating and confluent. However, termination is not modular in
general, even when combining two systems having no symbols in common
\cite{toyama87ipl}.

There exists many results on the modularity of confluence and
termination of first-order rewriting systems when these systems have
no symbols in common, or just undefined symbols (see for instance
\cite{gramlich12tcs} for some survey). But, here, we have higher-order
rewriting rules that may share defined symbols.

So, instead, we may try to apply general modularity results on
abstract relations \cite{doornbos98jigpal}. In particular, for all
terminating relations $P$ and $Q$, $P\cup Q$ terminates if $P$ steps
can be postponed, that is, if ${PQ}\sle{QP^*}$. In our case, we may
try to postpone the $\cD$ steps:

\begin{lemma}\label{lem-postpone}
  For all sets of higher-order rewriting rules $\cR$ and $\cD$, we
  have that ${\hookrightarrow_\beta}\cup{\hookrightarrow_\cR}\cup{\hookrightarrow_\cD}$ terminates if:
  \begin{enumerate}[(a)]\itemsep=0mm
  \item\label{pp-sn} ${\hookrightarrow_\beta}\cup{\hookrightarrow_\cR}$ and $\hookrightarrow_\cD$ terminate,
  \item\label{pp-ll} $\cR$ is left-linear,
  \item\label{pp-cl} $\cD$ is closed,
  \item\label{pp-nf} no right-hand side of $\cD$ is $\beta\cR$-reducible
    or headed by an abstraction,
  \item\label{pp-ov} no right-hand side of $\cD$ unifies with a non-variable
    subterm of a left-hand side of $\cR$.
  \end{enumerate}
\end{lemma}

\begin{proof}
  As usual, we define positions in a term as words on $\{1,2\}$:
  $\Pos(s)=\Pos(x)=\Pos(f)=\{\vep\}$, the empty word representing the
  root position, and $\Pos(tu)=\Pos(\lambda x:t,u)=\Pos(\Pi
  x:t,u)=1\cdot\Pos(t)\cup 2\cdot\Pos(u)$.
  
  Assume that $t\hookrightarrow_\cD u$ at position $p$ and $u\hookrightarrow_{\beta\cR}v$ at
  position $q$. If $p$ and $q$ are disjoint, then these reductions can
  be trivially permuted: $t\hookrightarrow_{\beta\cR}\hookrightarrow_\cD v$. The case $p\le q$ ($p$
  prefix of $q$) is not possible since $\cD$ is closed (\ref{pp-cl})
  and no right-hand side of $\cD$ is $\beta\cR$-reducible
  (\ref{pp-nf}). So, we are left with the case $q<p$:

  \begin{itemize}
  \item Case $u\hookrightarrow_\beta v$. The case $p=q1$ is not possible since no
    right-hand side of $\cD$ is headed by an abstraction
    (\ref{pp-nf}). So, $t|_q$ is of the form $(\lambda x:A,b)a$ and the
    $\cD$ step is in $A$, $b$ or $a$. Therefore, $t\hookrightarrow_\beta\hookrightarrow_\cD^*v$.
    
  \item Case $u\hookrightarrow_\cR v$, that is, when $u|_q=l\s$ where $l$ is a
    left-hand side of a rule of $\cR$. The case $p=qs$ where $s$ is a
    non-variable position of $l$ is not possible because no
    non-variable subterm of a left-hand side of $\cR$ unifies with a
    right-hand side of $\cD$ (\ref{pp-ov}). Therefore, since $l$ is
    left-linear (\ref{pp-ll}), $t|_q$ is of the form $l\theta$ for some
    substitution $\theta$, and the $\cD$ step occurs in some $x\theta$. Hence,
    $t\hookrightarrow_\cR\hookrightarrow_\cD^*v$.
  \end{itemize}
\end{proof}

\begin{example}
  As we have already seen, the typability conditions of
  $l=\mi{tail}~n~(\mi{cons}~x~p~v)$ is the set of equations
  $\cE=\{{\h{x}=T},{\h{p}=N},{\h{v}=Vp},{\h{n}=N},{V(sp)=V(sn)}\}$.
  By taking $\h{x}>\h{v}>\h{p}>\h{n}>V>T>N>s>p>n$ as total order on
  function symbols, the Knuth-Bendix completion procedure yields with
  $>_{lpo}$ the rewriting system $\cD=\{\h{x}\hookrightarrow
  T,\h{p}\hookrightarrow N,\h{v}\hookrightarrow Vp,\h{n}\hookrightarrow
  N,V(sp)\hookrightarrow V(sn)\}$. After Lemma \ref{lem-postpone},
  ${\hookrightarrow_\beta}\cup{\hookrightarrow_\cR}\cup{\hookrightarrow_\cD}$ is convergent if
  ${\hookrightarrow_\beta}\cup{\hookrightarrow_\cR}$ is convergent, $\cR$ is left-linear and $V$ and $s$
  are undefined. This works as well if, instead of $\cE$, we use its
  simplification
  $\cE'=\{{\h{x}=T},{\h{p}=N},{\h{v}=Vp},{\h{n}=N},{p=n}\}$. In this
  case, we get the rewriting system $\cD=\{\h{x}\hookrightarrow
  T,\h{p}\hookrightarrow N,\h{v}\hookrightarrow Vn,\h{n}\hookrightarrow
  N,p\hookrightarrow n\}$.
\end{example}

\begin{example}
  Finally, let's come back to the rewriting rule
  $\mr{app}\,a\,b\,(\mr{lam}\,a'\,b'\,f)\,x\hookrightarrow f\,x$ of
  Example \ref{ex-beta} encoding the $\beta$-reduction of simply-type
  $\lambda$-calculus. As already mentioned, the previous version of
  Dedukti was unable to prove that this rule preserves typing. Thanks
  to our new algorithm, the new version of
  Dedukti\footnote{\url{https://github.com/Deducteam/lambdapi}} can
  now do it.

  The computability constraints of the LHS are $\h{f}=\tau
  a'\rightarrow\tau b'$, $\tau a'\rightarrow\tau
  b'=\tau(\mr{arr}\,a\,b)$ and $\h{x}=\tau a$. Preservation of typing
  cannot be proved without simplifying this set of equations to
  $\h{f}=\tau a'\rightarrow\tau b'$, $\tau a'=\tau a$, $\tau b'=\tau
  b$ and $\h{x}=\tau a$.

  Then, any total order on function symbols allows to prove
  preservation of typing. For instance, by taking
  $\h{f}>\,\rightarrow\,>a'>a>b'>b$, we get the rewriting rules
  $\h{f}\hookrightarrow\tau a\rightarrow\tau b$, $\tau a'\hookrightarrow\tau
  a$, $\tau b'\hookrightarrow\tau b$ and $\h{x}\hookrightarrow\tau a$, so that
  one can easily check that, modulo these rewriting rules, $f\,x$ has
  type $\tau b$. Therefore,
  $\mr{app}\,a\,b\,(\mr{lam}\,a'\,b'\,f)\,x\hookrightarrow f\,x$ preserves
  typing.

  Note that the result does not depend on the total order taken on
  function symbols. For instance, if one takes
  $\h{f}>\,\rightarrow\,>a>a'>b'>b$ (flipping the order of $a$ and
  $a'$), we get the rewriting rules $\h{f}\hookrightarrow\tau
  a'\rightarrow\tau b$, $\tau a\hookrightarrow\tau a'$, $\tau
  b'\hookrightarrow\tau b$ and $\h{x}\hookrightarrow\tau a'$. In this case,
  $f\,x$ has type $\tau b$ as well. Flipping the order of $b$ and $b'$
  would work as well.
\end{example}

%%%%%%%%%%%%%%%%%%%%%%%%%%%%%%%%%%%%%%%%%%%%%%%%%%%%%%%%%%%%%%%%%%%%%%%%%%%%%%
\section{Related works and conclusion}
\label{sec-conclu}

The problem of type safety of rewriting rules in dependent type theory
modulo rewriting has been first studied for simply-typed function
symbols by Barbanera, Fern\'andez and Geuvers in
\cite{barbanera97jfp}. In \cite{blanqui05mscs}, the author extended
these results to polymorphically and dependently typed function
symbols, and showed that rule left-hand sides do not need to be
typable for rewriting to preserve typing. This was later studied in
more details and implemented in Dedukti by Saillard
\cite{saillard15phd}. In this approach, one first extracts a
substitution $\rho$ (called a pre-solution in Saillard's work) from the
typability constraints of the left-hand side $l$ and check that, if
$l$ is of type $A$, then the right-hand side $r$ is of type $A\rho$ (in
the same system). For instance, from the simplified set of constraints
$\cE'=\{{\h{x}=T},{\h{p}=N},{\h{v}=Vp},{\h{n}=N},{p=n}\}$ of our
running example, one can extract the substitution $\rho=\{(n,p)\}$ and
check that $v$ has type $(Vn)\rho=Vp$. However, it is not said how to
compute useful pre-solutions (note that we can always take the
identity as pre-solution). In practice, the pre-solution is often
given by the user thanks to annotations in rules. A similar mechanism
called inaccessible or ``dot'' patterns exists in Agda too
\cite{norell07phd}.

An inconvenience of this approach is that, in some cases, no useful
pre-solution can be extracted. For instance, if, in the previous
example, we take the original set of constraints
$\cE=\{{\h{x}=T},{\h{p}=N},{\h{v}=Vp},{\h{n}=N},{V(sp)=V(sn)}\}$
instead of its simplified version $\cE'$, then we cannot extract any
useful pre-solution.

In this paper, we proposed a more general approach where we check that
the right-hand side has the same type as the left-hand side modulo the
equational theory generated by the typability constraints of the
left-hand side seen as closed equations (Theorem \ref{thm-sr}). A
prototype implementation is available
on:\\\url{https://github.com/wujuihsuan2016/lambdapi/tree/sr}.

To ensure the decidability of type-checking in this extended system,
we propose to replace these equations by an equivalent but convergent
rewriting system using Knuth-Bendix completion \cite{knuth67} (which
always terminates on closed equations), and provide conditions for
preserving the termination and confluence of the system when adding
these new rules (Lemma \ref{lem-postpone}). This approach has also the
advantage that Dedukti itself can be used to check the type safety of
user-defined Dedukti rules.

We also showed that, for the algorithm to work, the typability
constraints sometimes need to be simplified first, using the fact that
some function symbols are injective (Theorem \ref{thm-sr2}). It would
be interesting to be able to detect or check injectivity automatically
(see \cite{wu19tr} for preliminary results on this topic), and also to
find a simplification procedure more general than the one of Figure
\ref{fig-simpl}.